\pdfoutput=1  
\documentclass[11pt]{article}

\usepackage[utf8]{inputenc}
\usepackage[T1]{fontenc}
\usepackage{amsmath,amssymb,amsthm,mathtools}
\usepackage[margin=0.9in]{geometry}
\usepackage[numbers,sort&compress]{natbib}
\usepackage{microtype}
\usepackage{enumitem}
\usepackage{booktabs}
\usepackage{graphicx}
\usepackage[font=small,skip=5pt]{caption}
\usepackage{hyperref}

\usepackage{titlesec}
\titlespacing*{\section}{0pt}{2.2ex plus 1ex minus .2ex}{1.2ex plus .2ex}
\titlespacing*{\subsection}{0pt}{1.8ex plus 1ex minus .2ex}{0.8ex plus .2ex}
\hypersetup{colorlinks=true,linkcolor=black,citecolor=black,urlcolor=black,
  pdftitle={Anonymous sharing is pairwise phase-blind: neutrality at order two
            and third-order structure in a fleet of checkpointing jobs},
  pdfauthor={Brieuc Le Roux Tardif},
  pdfsubject={Pulse-coupled oscillators, shared-resource contention,
              checkpoint storms},
  pdfkeywords={synchronisation, integrate-and-fire, processor sharing,
               checkpointing, datacenter power}}

\newcommand{\dd}{\mathrm{d}}

\theoremstyle{plain}
\newtheorem{theorem}{Theorem}
\newtheorem{proposition}[theorem]{Proposition}
\newtheorem{corollary}[theorem]{Corollary}
\newtheorem{lemma}[theorem]{Lemma}
\theoremstyle{definition}
\newtheorem{assumption}{Assumption}
\newtheorem{definition}[theorem]{Definition}
\newtheorem{remark}[theorem]{Remark}
\newtheorem{conjecture}[theorem]{Conjecture}

\title{\bfseries Anonymous sharing is pairwise phase-blind\\[2pt]
\large\mdseries Neutrality at order two and third-order structure in a fleet of
checkpointing jobs}

\author{Brieuc Le Roux Tardif\\
\small IMT Nord Europe\\
\small \texttt{brieuclerouxtardif@gmail.com}}

\date{\today}

\begin{document}
\maketitle

\begin{abstract}
Independent training jobs sharing a storage system write their checkpoints
through the same finite bandwidth, and the resulting bursts of correlated I/O,
together with the facility-level power transients that accompany them, are
commonly described as a self-reinforcing ``checkpoint storm''. We formalise the
self-reinforcement as phase locking in a population of integrate-and-fire
oscillators coupled through a shared resource, and show that within that model
it fails for a reason that has nothing to do with checkpointing. Call a resource
\emph{anonymous} if the rate it delivers to an active user depends on how many
users are active and not on which. For \emph{identical} jobs whose write is
shorter than their compute interval, an anonymous resource produces no pairwise
coupling at all: the two-job return map of the phase gap is the identity, under
storage contention, under a shared power cap and under both at once, so the
two-body interaction on which the Kuramoto and Mirollo--Strogatz frameworks are
built is not weak here but absent. Anonymity also freezes the firing order, for
any fleet size and any cap, so no trajectory reaches the synchronous state from
outside it. What survives is a third-order effect: where all $N$ write windows
overlap and the cap does not bind, the map is diagonal in the intervals between
consecutive write starts, $a_j \mapsto ((N-j)/j)\,a_j$, with reciprocal spectrum
and unit determinant, which makes synchrony a fixed point with
$\lceil N/2 \rceil - 1$ expanding directions rather than an attractor. That
determinant follows from anonymity and not from fairness: for any anonymous
throughput $f$ with $f(n) \le n$ the spectrum becomes $(N-j)f(j)/(j\,f(N-j))$,
whose product is still one. Numerically, a fleet launched at random neither
locks nor clusters, and the only memory of its launch our observables resolve
is that frozen order: after hundreds of cycles its configuration sits as far
from its own launch as from an independent one relabelled into the same order
sector, and no decay of its smallest gap is resolved. Absence of locking is not
absence of bursts, the upper tail of the number of concurrent writers staying
above its independent-phase value in every cell measured. Nothing erodes a
stagger in the deterministic
uncapped model; per-cycle jitter $\sigma$ does, a margin $m$ surviving
$0.11$ to $0.24\,(m/\sigma)^2$ cycles across the fleets tested. Heterogeneous
jobs behind a \emph{binding} cap do acquire a genuine pairwise coupling, which
is where the statement stops generalising.
\end{abstract}

\section{Introduction}
\label{sec:intro}

Large training runs checkpoint periodically, and a checkpoint is a burst: tens
to hundreds of gigabytes leave the accelerators for a parallel file system
while the GPUs sit idle. Two distinct things are called a \emph{checkpoint
storm}, and only one of them is our subject. Within a single distributed job,
thousands of ranks serialise their shards at the same instant by construction,
and the resulting burst is a design problem in the checkpointing stack itself
\citep{pytorch_dcp}. Across \emph{independent} jobs sharing one fabric, the
bursts contend without any mechanism aligning them, and the operational
question is whether they nonetheless drift into alignment. This paper is about
the second question only; nothing here bears on the first, which is a
sufficient cause of correlated I/O on its own. The electrical stakes are set by
the first: instantaneous fleet-level swings of tens of megawatts are reported
for a cluster of order $10^4$ accelerators \citep{llama3,semianalysis2025}, and
the sub-second timescale on which they matter for facility stability is the
subject of a growing power-engineering literature \citep{opalrt2026}.

The implicit model behind mitigation practice is that collisions between jobs
are self-reinforcing: jobs that once collide keep colliding, so schedules must
be deliberately staggered. Recent systems work treats the timing of a job's
synchronisation points as a schedulable quantity, choosing when to run DiLoCo
outer merges against measured fleet pressure \citep{diloco2026}; complementary
work smooths the power transients of a single job \citep{easyrider2026}.
Neither studies the free dynamics of the collision itself.

The question this raises is dynamical: \emph{is the synchronised state an
attractor?} If it is, staggering is futile, because contention drags jobs back
into phase. If it is repelling, contention is not the cause and a stagger has
nothing to undo, which is not the same as a fleet that can be left alone: three
properties have to be kept apart, asymptotic phase locking, transient
clustering, and the extreme upper tail of the number of concurrent writers. We
settle the first within the model, measure the second, and find the third
\emph{above} its independent-phase value even where the first two are absent
(\S\ref{sec:numerics}); \emph{storm} below always names the first, which is the
only one of the three the dynamical question is about. A companion paper \citep{lerouxtardif2026} asks the same
question for a different channel, a shared power envelope with a \emph{delayed}
throttle, and finds a coupling whose sign is set by the control lag: repulsive
to leading order, attractive only once the lag exceeds half a cycle. The present
paper removes the lag and the memory, keeps the sharing, and finds that the
pairwise term does not merely change sign, it vanishes. The two are reconciled
in \S\ref{sec:open}: the coupling in \citep{lerouxtardif2026} is a property of
the controller, not of the sharing.

The natural formalism is that of pulse-coupled oscillators: each job is an
integrate-and-fire unit whose phase is its progress through the
inter-checkpoint interval and whose discharge is the checkpoint write. For
excitatory coupling with a concave rise, Mirollo and Strogatz proved that
almost every initial condition synchronises \citep{mirollostrogatz1990}, and
inhibitory all-to-all coupling has since been shown to produce globally
attracting synchrony as well \citep{canavier2017}; applications have
concentrated on biological populations and on clock synchronisation in wireless
sensor networks \citep{wernerallen2005}. The datacenter version differs in the
coupling channel and, as we show, in the order of the interaction.

\paragraph{Contributions.} Graded by the strength of the supporting argument.
\begin{enumerate}[label=(\roman*),leftmargin=*,itemsep=1pt]
\item \textbf{(Theorem, proved)} For two \emph{identical} jobs with $d < T$,
the return map of the phase gap is the identity: under storage contention,
under a shared power cap, and under both at once, for any cap severity. The
pairwise coupling function is identically zero, not merely weak
(\S\ref{sec:pair}).
\item \textbf{(Theorem, proved)} For $N$ identical jobs whose write windows all
overlap, with the cap not binding, the return map is diagonal in interval
coordinates with eigenvalues $(N-j)/j$: reciprocal spectrum, unit determinant
on that branch, and synchrony a fixed point with $\lceil N/2 \rceil - 1$
expanding directions rather than an attractor. The determinant is a consequence
of anonymity and not of fairness, since an arbitrary anonymous throughput $f$
multiplies each eigenvalue by $f(j)/f(N-j)$ and leaves the product at one
(\S\ref{sec:nbody}).
\item \textbf{(Propositions, proved)} Two invariant structures
(\S\ref{sec:invariant}). The cyclic firing order is a constant of the motion,
for any fleet size, any cap and any anonymous rule, so distinct phases stay
distinct and synchrony is reached by no trajectory that does not start there.
And collision-free configurations, which exist exactly when $L \le N/(N-1)$,
are invariant and cannot be entered from outside, the flow being uniquely
reversible where the write windows are disjoint.
\item \textbf{(Numerical)} Off that set, nearby trajectories separate
exponentially in most launches, yet a single trajectory does not cluster: no
decay of the smallest gap is resolved over $800$ cycles, the gap of the pair
tightest at launch \emph{opens} by a factor $6$ to $33$, and Daido moments up to
$m = 4$ show no coherence that a change of null does not remove. Nor does it
keep of its launch anything these observables resolve beyond the invariant
order: its distance to the launch, $7$--$14\%$ of a cycle, is the distance to an
independent launch relabelled into the same order sector, cell by cell and
within a standard error, and the two distributions agree to a $1$-Wasserstein
distance below $5\%$ of their own spread, while the $22$--$26\%$ separating two
unrelabelled configurations measures a permutation the dynamics cannot perform.
The upper tail of the number of concurrent writers is nonetheless heavier than
that of independent phases in every cell (\S\ref{sec:numerics}).
\item \textbf{(Numerical, operational)} With per-cycle jitter $\sigma$ on the
compute time, each gap of a staggered fleet is a driftless random walk, and a
margin $m$ survives $0.11$ to $0.24\,(m/\sigma)^2$ cycles depending on the
fleet, verified over two decades (\S\ref{sec:practice}).
\item \textbf{(Numerical, negative)} Heterogeneous jobs behind a binding cap
\emph{do} couple pairwise, with a phase response supported on gaps below the
detuning, although both resources remain anonymous. This bounds the reach of
(i) (\S\ref{sec:limits}).
\end{enumerate}

\paragraph{Scope.} The results are statements about a model, and the model is
deliberately minimal (\S\ref{sec:model}). We prove what follows from three
assumptions, we verify each proof against an exact event-driven integrator, and
we state explicitly which assumptions, when relaxed, restore a pairwise
coupling, in one case by exhibiting it (\S\ref{sec:limits}). We do not claim
that real fleets never synchronise; we claim that in this model symmetric
contention between identical jobs cannot be the cause.

\section{Model}
\label{sec:model}

\begin{definition}[Fleet]
A fleet is a set of $N$ jobs. Job $i$ alternates two phases. In the
\emph{compute} phase it must accumulate an amount $T_i > 0$ of work at the
instantaneous rate $v(t)$. In the \emph{write} phase it must drain a checkpoint
volume $V_i > 0$ at the instantaneous rate $r_i(t)$. Time is normalised so that
a job computing alone advances at unit rate, and bandwidth is normalised so
that a job writing alone drains at unit rate; hence $d_i := V_i$ is the
\emph{solo write duration} of job $i$. Job $i$ \emph{fires} at the instant
$s_i$ at which it enters its write phase, and the observable is the vector of
firing times modulo the cycle. Write $n_w(t)$ for the number of jobs writing at
time $t$ and $n_c(t)$ for the number computing.
\end{definition}

\begin{assumption}[Blocking checkpoint]
\label{as:blocking}
A writing job does not compute: the phases are exclusive.
\end{assumption}

\noindent Justification: synchronous checkpointing, in which the training loop
halts while state is serialised and flushed, is the baseline that asynchronous
mechanisms are designed to improve on \citep{pytorch_dcp}, and it remains in use
wherever that path is unavailable or not enabled. We do not claim it is the
majority configuration in production, having no measurement to support such a
claim. This is the assumption that matters most (\S\ref{sec:limits}).

\begin{assumption}[Anonymous equal sharing]
\label{as:sharing}
Each writer receives bandwidth $f(n_w(t))/n_w(t)$, with $f$ positive,
$f(1) = 1$ and $f(n) \le n$; each computing job advances at
$v(t) = \min\{1, C/n_c(t)\}$, where $C \in (0,\infty]$ is a power cap expressed
as the number of jobs the cap sustains at full speed. Both rates depend on the
number of active users only, and are identical across active users. The
work-conserving case is $f \equiv 1$.
\end{assumption}

\noindent Justification: processor sharing is the standard idealisation of a
parallel file system serving concurrent streams without priorities, and of a
power budget divided uniformly across a homogeneous fleet. The condition
$f(n) \le n$ says only that concurrency makes no single stream faster than it
would be alone, and is used solely to order events in Theorem~\ref{th:nbody}.

\begin{assumption}[Deterministic per-cycle work]
\label{as:determ}
$T_i$ and $V_i$ are constants of job $i$. Heterogeneity across jobs is
permitted; stochasticity within a job is not, except where \S\ref{sec:practice}
introduces it explicitly.
\end{assumption}

\noindent Justification: the compute time between checkpoints is set by a fixed
number of steps of a fixed model, and is repeatable to a few percent absent
failures.

\begin{definition}[Anonymous resource, phase-blind coupling]
\label{def:anon}
A shared resource is \emph{anonymous} if the rate it delivers to an active user
depends only on how many users are active, and not on which. Let $k$ units be
taken in isolation and let the section be the set of states at which every unit
is computing and the numbers of writes they have completed differ by a
prescribed offset, zero being the case of units launched together. There every
unit has the same mode and no residual volume, so the hybrid state is the vector
of residual works, which the flow translates uniformly for as long as all units
compute; the state modulo that translation is the vector of $k-1$ firing
differences, and the return map $F$ of those differences is well defined on it.
That the section is nonempty and recurrent with a constant offset is a property
of the system and not of the definition: it is proved for the pair of
Theorem~\ref{th:pair} under $d < T$ and assumed in
Conjecture~\ref{conj:blind}. Let
\begin{equation}
g(\delta) := F(\delta) - \delta
\label{eq:increment}
\end{equation}
be the increment that map applies. The coupling is \emph{phase-blind at order
$k$} if $g$ is constant, that is if the increment does not depend on the phase
differences themselves. Identical units then have $g \equiv 0$, since exchange
symmetry fixes the synchronous configuration and $g$ is constant; units that
differ are transported by a translation set by their detuning alone on a fixed
event-order branch with the cap not binding
(Proposition~\ref{prop:hetero}), and not in general (\S\ref{sec:hetero}).
Phase-blindness is the exact vanishing of the coupling, not of the frequency
spread.
\end{definition}

\noindent The hierarchy of orders is the hierarchy of subsystems: order $k$ is a
statement about $k$ units \emph{in isolation}, not a term in a series expansion.
Blindness stopping at order three therefore means that no dependence on relative
phase appears before three jobs are present, and Theorem~\ref{th:nbody} exhibits
it at exactly three.

\paragraph{Dimensionless parameters.} Set $T_i \equiv 1$ for the homogeneous
case. The two parameters that survive are the \emph{duty} $d$, the solo write
duration in units of the compute time, and the \emph{aggregate demand}
$L := \sum_i d_i = Nd$, the write volume the fleet requests per unit of compute
time. Neither is a duty cycle in the usual sense: a job alone occupies the
fabric for a fraction $d/(1+d)$ of its free period $P_0 = 1 + d$. A cyclic
schedule giving every write an interval of at least $d$ to itself exists
exactly when
\begin{equation}
Nd \le 1 + d ,
\qquad\text{equivalently}\qquad
L \le L^\star := \frac{N}{N-1} ,
\label{eq:sat}
\end{equation}
which tends to $1$ from above as the fleet grows. We call $L^\star$ the
\emph{stagger-feasibility threshold} rather than a saturation threshold: it is
a geometric condition on the gaps, and it is separately observed
(\S\ref{sec:trap}) that the effective cycle begins to stretch beyond $P_0$ once
it is crossed. The same threshold governs Proposition~\ref{prop:frozen}.

\section{Order two: anonymous resources are phase-blind}
\label{sec:pair}

The argument rests on a single accounting identity, and the quantity it
conserves is not the phase gap itself but the difference in accumulated work.
Write $v_1 := \min\{1, C\}$ for the compute rate of a job that computes while
the other writes, and $v_2 := \min\{1, C/2\}$ for the rate when both compute;
$v_1 = v_2 = 1$ exactly when $C \ge 2$. Let $W_i, K_i \subset [0,t]$ be the sets
of times at which job $i$ writes and computes. Assumption~\ref{as:blocking}
makes them a partition of $[0,t]$, which is the hypothesis that carries
everything below:
\begin{equation}
|K_i| = t - |W_i| .
\label{eq:partition}
\end{equation}
With two jobs, ``job $i$ writes and job $j$ does not'' means job $i$ is the
\emph{only} writer, so it drains at $f(1) = 1$; likewise a job computing while
the other writes is the only computer and advances at $v_1$. An interval on which
one job is alone is one on which its rate is known, and anonymity makes that rate
the same for whichever job is alone.

\begin{lemma}[Equal occupancy, then equal solo work]
\label{lem:deficit}
Let $N=2$ with identical jobs ($T_i \equiv T$, $d_i \equiv d$) under
Assumptions~\ref{as:blocking}--\ref{as:sharing}, and let $t$ be an instant at
which both jobs are computing, having completed $k_1$ and $k_2$ writes
respectively and started no further one. Then, for any cap $C > 0$,
\begin{equation}
|W_1| - |W_2| = (k_1 - k_2)\,d , \qquad
|K_1 \setminus K_2| - |K_2 \setminus K_1| = (k_2 - k_1)\,d .
\label{eq:occupancy}
\end{equation}
Consequently the two jobs have accumulated the same amount of compute work
outside the intervals on which both compute, up to a term set by the write
offset alone: writing $w_i(t)$ for the compute work job $i$ has accumulated
since time $0$, the difference $\Delta w(t) := w_1(t) - w_2(t)$ depends on $t$
only through $k_1 - k_2$. Two counters have to be kept apart. Decompose
$w_i(t) = k_i T + x_i(t)$, with $x_i(t) \in [0, T)$ the work banked in the
current cycle, which is what the next firing time is set by; then along a
sequence of section instants of \emph{constant} offset $k_1 - k_2$ the two
differences $\Delta w$ and $\Delta x := x_1 - x_2$ are constant together, and
$|\Delta x| < T$ whatever the offset.
\end{lemma}

\begin{proof}
\emph{Writing.} By hypothesis job $i$ has drained exactly $k_i d$ by time $t$,
with no partial write outstanding. Split that volume over $W_1 \cap W_2$ and
$W_i \setminus W_j$. On the common part both writers receive the same rate,
whatever it is, so the volumes drained there are equal; call the common value
$S$. On $W_i \setminus W_j$ job $i$ is the only writer and drains at $f(1) = 1$,
so it drains $|W_i \setminus W_j|$ there. Hence $k_i d = S + |W_i \setminus W_j|$
for $i = 1, 2$, so $|W_1 \setminus W_2| - |W_2 \setminus W_1| = (k_1-k_2)d$ and
therefore $|W_1| - |W_2| = (k_1-k_2)d$.

\emph{Computing.} The partition \eqref{eq:partition} turns that into
$|K_1| - |K_2| = (k_2-k_1)d$, and subtracting the common part $|K_1 \cap K_2|$
gives the second display of \eqref{eq:occupancy}. On $K_i \setminus K_j$ job $i$
computes while job $j$ writes, so it is the only computer and advances at $v_1$;
on $K_1 \cap K_2$ both advance at the same rate and accumulate equally. Every
contribution to $w_1 - w_2$ after time $0$ therefore cancels except the solo
term, leaving $\Delta w(t) = \Delta w(0) + v_1 (k_2 - k_1) d$, a function of the
offset alone.
\end{proof}

\begin{theorem}[Pair neutrality]
\label{th:pair}
Let $N = 2$ with identical jobs under
Assumptions~\ref{as:blocking}--\ref{as:determ}, with $d < T$, with any
combination of the two shared resources and any cap $C>0$. Let $\delta_k$ be the
firing gap on cycle $k$. Then $\delta_{k+1} = \delta_k$ for every $k$ and every
$\delta_0$: the increment \eqref{eq:increment} vanishes identically, and no
phase gap is either attracted or repelled.
\end{theorem}

\begin{proof}
The section instants of Lemma~\ref{lem:deficit} recur, and establishing that
needs no control over the effective period, which a binding cap makes several
times $P_0$. Suppose, for contradiction, that the two jobs never compute
simultaneously on $[0,t]$. Assumption~\ref{as:blocking} makes the phases
exclusive, so job $1$ computing forces job $2$ to be writing: $K_1 \subseteq W_2$,
and on $K_1$ job $2$ is then the \emph{only} writer, hence drains at $f(1) = 1$.
The volume it drains on $K_1$ is therefore $|K_1|$, which cannot exceed the
volume it has drained by time $t$, itself at most $(k_2+1)d$ with $k_2$ its
number of completed writes; and job $1$ has accumulated $k_1 T$ of work at a rate
at most $1$, so $|K_1| \ge k_1 T$. Hence $k_1 T \le (k_2+1)\,d$, and
symmetrically $k_2 T \le (k_1+1)\,d$. Every rate is bounded below, by
$\min\{1, C/2\}$ in the compute phase and by $f(2)/2$ in the write phase, so
$k_1$ and $k_2$ grow without bound with $t$; summing the two inequalities gives
$T \le d\,\bigl(1 + 2/(k_1+k_2)\bigr)$ and hence $T \le d$ in the limit, against
the hypothesis $d < T$. The two jobs therefore compute simultaneously at some
instant, and the argument applied from that instant onward makes such instants
recur. That consecutive ones are separated by exactly one write of each job,
which is what keeps the offset $k_1 - k_2$ constant along the sequence and lets
Lemma~\ref{lem:deficit} return one and the same $\Delta x$ at all of them, is
established at the end of the proof by following one cycle through.

Let $t$ be such an instant and suppose $\Delta x := \Delta x(t) > 0$, so job $1$
leads. From $t$ both jobs compute, at the common rate $v_2$, until job $1$
completes its remaining work $T - x_1$ and fires; at that moment job $2$ still
needs exactly $\Delta x$. It accumulates that residue while job $1$ writes, alone
and therefore at $v_1$, for at most the solo write duration $d$; if it has not
fired by then, job $1$ resumes computing and the rate returns to $v_2$. The
firing gap is therefore
\begin{equation}
\delta_{k+1} = \theta(\Delta x) :=
\begin{cases}
\Delta x / v_1, & \Delta x \le v_1 d, \\[2pt]
d + (\Delta x - v_1 d)/v_2, & \Delta x > v_1 d,
\end{cases}
\label{eq:theta}
\end{equation}
a strictly increasing function of $\Delta x$ alone. Lemma~\ref{lem:deficit}
makes $\Delta x$ the same at every section instant, so $\delta_{k+1}$ is the
same on every cycle, and in particular equals $\delta_k$.

It remains to follow each case to the end of the cycle, which is what shows that
neither job gains a write on the other. If $\Delta x \le v_1 d$, job $2$ fires
while job $1$ is still writing, job $1$ having drained $\Delta x / v_1 \le d$ by
then. From that instant both write and receive the same rate, so job $1$ keeps
its lead $\Delta x / v_1$ in drained volume and finishes first, leaving job $2$
exactly that residue to drain alone at rate $1$; job $1$ computes meanwhile, and
is the only computer, so it accumulates $v_1 \cdot \Delta x / v_1 = \Delta x$ by
the time job $2$ finishes. If $\Delta x > v_1 d$, job $1$ drains alone and
finishes in $d$, both compute at $v_2$ from then on, and job $2$ fires first
since it needs $(\Delta x - v_1 d)/v_2 < T/v_2$ against the $T/v_2$ job $1$ needs
from an empty accumulator; job $2$ then writes alone for exactly $d$ while job
$1$ computes alone at $v_1$, so job $1$ holds $(\Delta x - v_1 d) + v_1 d
= \Delta x$ when that write ends. In both cases job $1$ has banked $\Delta x$
in the new cycle, which is smaller than $T$ by Lemma~\ref{lem:deficit}, so it has
not fired again; job $2$ has banked nothing, and the instant at which job $2$'s
write ends is a section instant at which each job has completed exactly one
further write, so the offset is unchanged and the current-cycle difference is
again $\Delta x$.
\end{proof}

\begin{remark}[On the hypothesis $d < T$, and on the conserved quantity]
\label{rem:dT}
The restriction is used only to make the section nonempty, and $d < T$ is the
regime of interest, a checkpoint being shorter than the compute interval it
protects. Whether the conclusion survives $d \ge T$ we do not know: the
integrator returns the same exact neutrality at $d = 1.5$ and $d = 2.5$ with
$T = 1$ at every cap tested, which we report as evidence and not as a theorem.
The invariant, in either case, is the current-cycle work difference $\Delta x$
and not the firing gap, of which it is the preimage under \eqref{eq:theta}. The
distinction is invisible when $C \ge 2$, where $\theta$ is the identity, and it
is why a proof treating the isolated firing time as $s_i + T + d$ fails for
$C < 1$, where no job ever runs at unit rate and that reference schedule is not
the free one; the same caution rules out bounding an event separation by a
fraction of $P_0$, the measured period reaching $3.8\,P_0$ at $C = 0.5$.
\end{remark}

\begin{remark}[What the proof uses]
\label{rem:robust}
No step used the value of any rate, only that two jobs in the same phase receive
the same one and that a job alone in a phase receives a rate independent of its
identity. Neutrality therefore holds for a cap so severe that a single job
cannot run at nominal speed ($C < 1$), a regime in which the instantaneous rates
of the two jobs do \emph{not} coincide: equality of \emph{occupancy} restores
over a cycle what the rates break at each instant (Table~\ref{tab:p1}). The
proof also survives letting writers draw on the power budget,
$n_c \to n_c + \alpha n_w$ with $\alpha \ge 0$, which changes $v_1$ in
\eqref{eq:theta} and nothing else. What it does not survive is unequal
allocation, nor, once the cap binds, unequal jobs (\S\ref{sec:limits}).
\end{remark}

\begin{proposition}[Heterogeneous jobs, cap not binding: local translation on a
fixed event-order branch]
\label{prop:hetero}
Let two jobs differ in write volume or compute work, with $C \ge 2$, and let
$\delta$ be the firing gap measured on the section at which the leading job
fires. On a branch along which the event ordering is unchanged,
\begin{equation}
\delta' = \delta + (d_2 - d_1) + (T_2 - T_1) ,
\label{eq:pairmap}
\end{equation}
a translation with no dependence on $\delta$: the paired firing times advance by
$P_2 - P_1$ per cycle, with $P_i := T_i + d_i$. That is the difference of the two
free periods and not a claim that either job runs at its own: where the branch
hypothesis fails, below, neither does. Two
conventions have to be fixed for this to mean anything. The labels are the jobs
themselves and not their ranks, so they do not exchange when the lagging job
overtakes the leading one; and $\delta$ is not reduced modulo a period, there
being two periods to choose between, so it grows without bound. The hypothesis
on the event ordering is not vacuous, and what happens when it fails is
instructive: at $d_1 = 0.1$, $d_2 = 0.3$ the writes overlap on $24$ of the $106$
cycles, and on those the two jobs are delayed by the same amount, exactly as
Theorem~\ref{th:pair} requires, but not on the same indexed cycle, their periods
differing. Pairing firings by index therefore moves the increment to $0.1$ or
$0.3$ on those cycles while its average over the run stays at $P_2 - P_1$.
\end{proposition}

\begin{proof}
Rerun the proof of Lemma~\ref{lem:deficit} with $d_1 \ne d_2$. The common
drained volume $S$ is still the same for both, so $d_i = S + |W_i \setminus W_j|$
now gives $|W_2| - |W_1| = d_2 - d_1$: the write phase transports the lag and
adds the volume difference. With $C \ge 2$ every computing job advances at unit
rate whether alone or not, so $v_1 = v_2 = 1$, the map \eqref{eq:theta} is the
identity, and the compute phase likewise transports the lag and adds
$T_2 - T_1$. Summing the two gives \eqref{eq:pairmap}.
\end{proof}

\noindent The hypothesis $C \ge 2$ is not cosmetic. When the cap binds,
$v_1 \ne v_2$, the two jobs no longer convert a work lead into a time lag at the
same rate, and a heterogeneous pair acquires a genuine coupling even though both
resources stay anonymous. That is the counterexample of \S\ref{sec:limits}, and
it is the sharpest limit on how far this section generalises.

The interpretation is the point of departure for everything that follows. Both
Kuramoto theory and Mirollo--Strogatz theory are built on a pairwise
interaction, a coupling function $\Gamma(\theta_i - \theta_j)$ whose shape
decides whether the synchronous state attracts, and for two identical jobs that
object is identically zero: contention charges them \emph{equal occupancy},
which transports their work difference unchanged, and a preserved work
difference translates the pair rather than rotating its relative phase.

\section{Order three: the \texorpdfstring{$N$}{N}-writer return map}
\label{sec:nbody}

\begin{theorem}[Diagonal return map]
\label{th:nbody}
Let $N$ identical jobs ($T_i \equiv 1$, $d_i \equiv d$) under a cap $C \ge N$,
that is under storage contention alone, fire at times
$0 = s_1 < s_2 < \cdots < s_N < d$, so that all $N$ write windows overlap. Write
$a_j := s_{j+1} - s_j$ for $j = 1, \dots, N-1$. Then, with work-conserving
sharing, the intervals between consecutive firings on the next cycle are
\begin{equation}
a_j' = \frac{N-j}{j}\, a_j , \qquad j = 1, \dots, N-1 .
\label{eq:nmap}
\end{equation}
The map is diagonal, its spectrum is $\{\lambda_j = (N-j)/j\}_{j=1}^{N-1}$, and
$\prod_{j} \lambda_j = 1$.
\end{theorem}

\begin{proof}
First, the hypothesis $s_N < d$ does force all $N$ windows to overlap: by
$f(n) \le n$ a job drains at rate at most $1$, so no job can finish before a
time $d$ has elapsed since its own start, and in particular
$t_1 \ge s_1 + d > s_N$. Hence every job has started, and none has finished, at
time $s_N$.

Since the drain rate is $1/n_w$ and is common to all active writers, the
\emph{difference} in drained volume between two jobs that are both writing is
constant in time. On $[s_j, s_{j+1})$ exactly $j$ jobs write, namely
$1, \dots, j$, none of which has finished by the previous paragraph, so job $j$
drains $a_j / j$ while job $j+1$ has not started; from $s_{j+1}$ onward both are
active whenever both are unfinished, so job $j$ retains the lead
\begin{equation}
\Lambda_j = \frac{a_j}{j}
\label{eq:lead}
\end{equation}
until it finishes. The leads are strictly positive, so jobs finish in their
firing order, $t_1 < t_2 < \cdots < t_N$. At $t_j$, job $j$ has drained $d$,
hence job $j+1$ has drained $d - \Lambda_j$ and has $\Lambda_j$ remaining. On
$[t_j, t_{j+1})$ the writers are exactly the jobs $j+1, \dots, N$, that is
$n_w = N - j$, so job $j+1$ drains at rate $1/(N-j)$ and
$t_{j+1} - t_j = (N-j)\,\Lambda_j$. Each job then computes for one unit of time:
$n_c \le N \le C$ at every instant, so $\min\{1, C/n_c\} = 1$ and the compute
duration is the same for all jobs whatever the order in which they finish
writing. This is the only step that uses $C \ge N$, and it fails as soon as the
cap binds. The next firing times are therefore $t_j + 1$ and the new intervals
are $a_j' = t_{j+1} - t_j$, which is \eqref{eq:nmap}. Finally
$\prod_{j=1}^{N-1} (N-j)/j = (N-1)!/(N-1)! = 1$.
\end{proof}

\begin{corollary}[Reciprocal spectrum; synchrony is a fixed point, not an
attractor]
\label{cor:spec}
The map \eqref{eq:nmap} preserves phase-space volume on its branch and its
spectrum is reciprocal, $\lambda_j \lambda_{N-j} = 1$, with
$\lceil N/2 \rceil - 1$ expanding directions and as many contracting ones; for
odd $N$ these exhaust the spectrum, while for even $N$ the mode $j = N/2$ has
$\lambda = 1$. The synchronous configuration $a = 0$ is the apex of the closed
cone $\{a_j \ge 0\}$ on which \eqref{eq:nmap} is linear and is a fixed point of
it, not an interior point of the branch, whose hypothesis is a strict ordering.
The statement is therefore directional rather than a linearisation at a smooth
point: a perturbation of synchrony falls in one of the $N!$ cones indexed by the
order in which the jobs fire inside the burst, $N$ of which share each of the
$(N-1)!$ cyclic orders that Proposition~\ref{prop:order} conserves; relabelling
identical jobs exchanges those cones,
and on each of them the map is \eqref{eq:nmap}. Since $\lambda_1 = N - 1 > 1$
for $N \ge 3$, every perturbation with $a_1 > 0$ is amplified and no
neighbourhood of $a = 0$ is contracted; for $N = 2$ the map is neutral,
consistently with Theorem~\ref{th:pair}. More generally a cluster state, in which
$a_j = 0$ for the ranks inside each cluster, is invariant by exchange symmetry,
and the cluster occupying ranks $j, j+1$ tightens if $j > N/2$ and disperses if
$j < N/2$. What $\det = 1$ excludes is a fixed point of this branch map that is
asymptotically stable in \emph{all} of its directions; it does not exclude
attraction to a cluster, a manifold of lower dimension being able to attract
transversally while expanding along itself, and Corollary~\ref{cor:recruit}
exhibits exactly that.
\end{corollary}

\begin{proof}
Everything except invariance is read off the spectrum of \eqref{eq:nmap}, which
is a linear diagonal map on the branch and therefore has $a = 0$ as a fixed
point with those eigenvalues. Invariance of a cluster state is exchange
symmetry: identical jobs in identical states receive identical rates at every
instant, so their states coincide for all time.
\end{proof}

\begin{corollary}[Recruitment on an invariant cluster manifold]
\label{cor:recruit}
Let $\mathcal{M}_1 := \{a_1 = 0\}$ be the manifold on which the two leading jobs
fire together. It is invariant, and the restriction of \eqref{eq:nmap} to it
acts on the surviving intervals with the same eigenvalues. For $N = 3$ that
restriction is $a_2' = a_2/2$: the third job is recruited by the pair
geometrically and the fleet tends to full synchrony, in infinite time and
without any two jobs ever coinciding. Synchrony is therefore asymptotically
stable \emph{relative to} $\mathcal{M}_1$ while being unstable in the ambient
space, which a unit determinant does not forbid. Two things confine the
phenomenon and one does not. The transverse direction carries
$\lambda_1 = N - 1 > 1$, so $\mathcal{M}_1$ repels on this branch; and by
Proposition~\ref{prop:order} no trajectory whose jobs fire at distinct instants
\emph{enters} it, so the recruitment lives on a set of measure zero the dynamics
cannot reach in finite time. What neither argument excludes is $a_1^{(k)} > 0$ at
every $k$ with $a_1^{(k)} \to 0$, an asymptotic approach to a cluster along the
other branches of the global map, on which we compute nothing;
\S\ref{sec:order} bounds that by measurement in four cells, not by proof.
Symmetrically, the trailing
manifolds $\{a_j = 0\}$ with $j > N/2$ are transversally contracting on this
branch, but $\lambda_1 > 1$ carries any trajectory with $a_1 > 0$ off the branch
in finitely many cycles (Remark~\ref{rem:domain}), so that contraction is not
asymptotic on this branch either.
\end{corollary}

\begin{proof}
The map is diagonal, so $a_1 = 0$ gives $a_1' = 0$ and leaves the other
coordinates as in \eqref{eq:nmap}; invariance is again exchange symmetry. For
$N = 3$ the surviving eigenvalue is $\lambda_2 = 1/2$, so $\sum_j a_j$ decreases
and the branch hypothesis $\sum_j a_j < d$ of Remark~\ref{rem:domain} is
preserved at every cycle, which makes the iteration legitimate for all time and
gives $a_2^{(k)} = 2^{-k} a_2^{(0)}$. The integrator returns that ratio to
machine precision from $a_1 = 0$, $a_2 = 0.1$, $d = 0.3$.
\end{proof}

\begin{remark}[What ``volume preserving'' does and does not say]
\label{rem:branchdet}
Both $\det = 1$ and the spectrum are properties of the single branch on which
all $N$ windows overlap and the cap does not bind. The global map is piecewise
affine with many branches; we compute the determinant on no other, and do not
claim the pieces glue into a globally measure-preserving map, which where the cap
binds they demonstrably do not (\S\ref{sec:limits}). Nothing in
\S\ref{sec:numerics} rests on this corollary.
\end{remark}

\begin{proposition}[Volume preservation is anonymity, not fairness]
\label{prop:anon}
Let the writers share an anonymous storage rule of total throughput $f(n_w)$,
each active writer receiving $f(n_w)/n_w$, with $f$ positive and $f(n) \le n$.
Under the hypotheses of Theorem~\ref{th:nbody} the map is still diagonal, with
\begin{equation}
a_j' = \frac{(N-j)\,f(j)}{j\,f(N-j)}\, a_j .
\label{eq:anonmap}
\end{equation}
Its spectrum is again reciprocal, $\lambda_j \lambda_{N-j} = 1$, and
$\det = 1$ for every such $f$.
\end{proposition}

\begin{proof}
The condition $f(n) \le n$ keeps every individual rate at most $1$, so the
event ordering established in the first paragraph of the proof of
Theorem~\ref{th:nbody} is unchanged. The rate then enters in exactly two places.
On $[s_j, s_{j+1})$ there are $j$ writers, so job $j$ drains at $f(j)/j$ and
banks $\Lambda_j = a_j f(j)/j$; on $[t_j, t_{j+1})$ there are $N-j$, so job
$j+1$ cashes that residue out at $f(N-j)/(N-j)$, taking
$t_{j+1} - t_j = \Lambda_j (N-j)/f(N-j)$, which is \eqref{eq:anonmap}. In
between, both jobs receive the same rate whenever both are unfinished, whatever
that rate is, so the lead is unchanged. Anonymity is the only property of $f$
used. Reciprocity is immediate, and $\prod_j \lambda_j = 1$ follows by pairing
$j$ with $N-j$.
\end{proof}

\begin{remark}[Domain of validity]
\label{rem:domain}
Equation \eqref{eq:nmap} holds for one cycle of a configuration satisfying
$s_N < d$, and may be iterated only while $\sum_j a_j' < d$. Since
$\lambda_1 = N-1 > 1$ for $N \ge 3$, the leading interval grows geometrically
and that condition fails after finitely many cycles: the fully overlapping
region is transient, and the spectrum describes one branch of a piecewise-affine
map rather than the asymptotics, which \S\ref{sec:numerics} measures instead.
\end{remark}

\begin{figure}[t]
\centering
\includegraphics[width=0.78\textwidth]{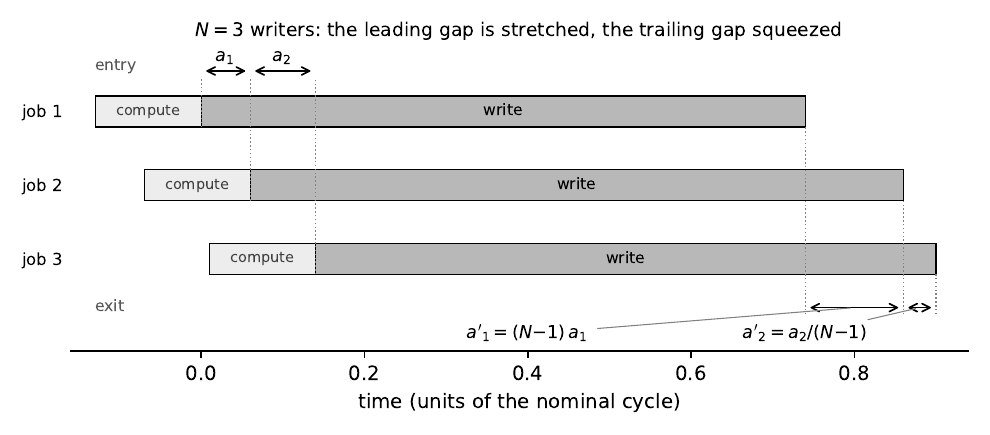}
\caption{The mechanism of Theorem~\ref{th:nbody} for $N=3$. Jobs enter their
write phase at intervals $a_1, a_2$; because a job that entered earlier banked a
lead at a lower level of concurrency and cashes it out at a higher one, the exit
intervals are $a_1' = (N-1) a_1$ and $a_2' = a_2/(N-1)$: the leading gap is
stretched, the trailing gap squeezed, and their product preserved.}
\label{fig:mechanism}
\end{figure}

Two readings of \eqref{eq:nmap} matter physically. \emph{The head disperses}:
$a_1' = (N-1) a_1$, so a job firing slightly ahead of a pack is ejected from it.
\emph{The tail compacts}: $a_{N-1}' = a_{N-1}/(N-1)$, so groups formed at the
back of the burst tighten. Equation \eqref{eq:lead} also shows where the pairwise
cancellation of \S\ref{sec:pair} fails: the lead $\Lambda_j$ is banked at $1/j$,
set by how many \emph{other} jobs were already writing, and cashed out at
$1/(N-j)$, set by how many others still are, both third-party counts. With
$N = 2$ they degenerate to $j = N - j = 1$ and \eqref{eq:nmap} returns
$a_1' = a_1$, recovering Theorem~\ref{th:pair}.

\section{Two invariant structures}
\label{sec:invariant}

\begin{proposition}[The firing order is invariant]
\label{prop:order}
Let $N$ identical jobs run under Assumptions~\ref{as:blocking}--\ref{as:determ},
with any cap $C > 0$ and any anonymous throughput rule. If $s_i^k < s_j^k$ then
$s_i^{k+1} < s_j^{k+1}$. The cyclic firing order is therefore a constant of the
motion, all the intervals $a_j$ keep their sign, and no two jobs ever exchange
rank or coincide.
\end{proposition}

\begin{proof}
Two steps, each an application of anonymity. \emph{Write phase.} Suppose
$s_i^k < s_j^k$. On $[s_i^k, s_j^k)$ job $i$ drains at a strictly positive rate
and job $j$ has not started, so $i$ acquires a strictly positive lead in drained
volume. Whenever both are writing they receive the same rate, so that lead is
constant; whenever $j$ writes and $i$ does not, $i$ has already finished. In
either case $e_i^k < e_j^k$. \emph{Compute phase.} Two cases, since job $i$ may
reach $T$ before job $j$ has finished writing. If it does, it fires at
$s_i^{k+1} < e_j^k < s_j^{k+1}$, job $j$ having to finish write $k$ and then a
compute phase before firing again, and there is nothing to prove: this is the
only step at which one job laps the other. Otherwise $i$ computes throughout
$[e_i^k, e_j^k)$, at the rate $v(t) = \min\{1, C/n_c(t)\} > 0$, while job $j$ is
still writing and by Assumption~\ref{as:blocking} accumulates no work; $i$ therefore
holds a strictly positive lead in accumulated work at $e_j^k$. From then on both
receive the same $v(t)$ whenever both compute, so the lead is preserved until
$i$ reaches $T$ and fires. Hence $s_i^{k+1} < s_j^{k+1}$.
\end{proof}

\begin{corollary}[Synchrony is invariant but unreachable]
\label{cor:unreachable}
The synchronous set, and more generally each cluster manifold, is invariant
(Corollary~\ref{cor:spec}) and of zero Lebesgue measure, and by
Proposition~\ref{prop:order} it is reached by no trajectory that does not start
in it. A fleet whose jobs fire at distinct instants keeps them distinct for all
time; clustering, if it occurred, could only be asymptotic.
\end{corollary}

\begin{proposition}[Collision-free configurations: invariant, unreachable, of
computable measure]
\label{prop:frozen}
Let $N$ identical jobs run under a cap $C \ge N$, and let $\mathcal{C}$ be the
set of configurations whose $N$ cyclic gaps between write starts are all at
least $d$. On $\mathcal{C}$ no two write windows overlap, every job runs at unit
rate in both phases, every cycle lasts exactly $P_0 = 1 + d$, and every gap is
constant for all time. The set is nonempty if and only if $Nd \le 1 + d$; it is
closed, and its interior is nonempty exactly when the inequality is strict, with
Lebesgue measure
\begin{equation}
\mu(\mathcal{C}) = \Bigl(1 - \frac{Nd}{P_0}\Bigr)^{N-1}
\label{eq:frozenmeasure}
\end{equation}
relative to starts distributed uniformly over the cycle. Moreover no trajectory
enters the interior of $\mathcal{C}$ from outside it.
\end{proposition}

\begin{proof}
A job whose write starts at least $d$ after the previous one and at least $d$
before the next is the only writer for the whole of its write, so it drains at
$f(1) = 1$ and finishes in exactly $d$; and $n_c \le N \le C$ makes every
computing job advance at unit rate, so its compute lasts exactly $1$. Every job
therefore has period exactly $P_0$, the configuration advances rigidly, and the
hypothesis is reproduced at the next firing, so $F$ restricted to $\mathcal{C}$
is the identity and $F(\mathcal{C}) = \mathcal{C}$. For existence, the $N$ cyclic
gaps are nonnegative and sum to $P_0$, so all can be at least $d$ if and only if
$Nd \le 1+d$; the set they form is the closed simplex
$\{g_i \ge d, \sum_i g_i = P_0\}$, whose relative volume is
\eqref{eq:frozenmeasure} and whose interior is nonempty exactly when
$Nd < 1 + d$. Finally, at an interior state the write windows are pairwise
disjoint, so no two events coincide and every rate equals $1$; the evolution is
locally a translation at unit speed with isolated events, and integrating it
backwards from one section to the previous one gives a unique state, with the
same disjoint windows. Any $x$ with $F(x) \in \mathcal{C}$ is therefore that
state and is collision-free. On the boundary, where some gap equals $d$ exactly,
the backward continuation is the one the tie-breaking convention selects.
\end{proof}

\noindent Three consequences. A stagger in which every write has clearance $d$
is permanent, and exists exactly at and below the threshold \eqref{eq:sat}; the
margin available to a fleet spread evenly is the smaller quantity
$m = P_0/N - d$ of \S\ref{sec:practice}. A fleet that does not start
collision-free never reaches the interior of $\mathcal{C}$, in either time
direction, so the numerics below is entirely a statement about the contending
set, the boundary being left to the tie-breaking convention. And random launches
land in $\mathcal{C}$ with probability \eqref{eq:frozenmeasure}, which is
$9.2 \times 10^{-2}$ at $N = 8$, $L = 0.3$ and $9.2 \times 10^{-29}$ at
$N = 32$, $L = 0.9$: at fleet scale a collision-free schedule has to be
constructed, never encountered.

\section{Numerical study}
\label{sec:numerics}

\subsection{Protocol}

All results below come from an exact event-driven integrator
(\texttt{sim/exact.py}): between two events all rates are piecewise constant, so
stepping from event to event is exact up to floating-point round-off. This is
not a convenience, since a fixed-step integrator produces spurious drift of
order the step size, which is precisely the quantity under test. It verifies the
model as implemented, not the model against a facility; no quantity here is
fitted to a measured trace. Ties between simultaneous events are broken by
processing write completions before write starts, a convention that selects one
continuation on a set of measure zero, which is however exactly where the
synchronous and cluster manifolds and the boundary of $\mathcal{C}$ sit; the
results about those objects (Corollary~\ref{cor:spec},
Propositions~\ref{prop:order} and \ref{prop:frozen}) are proved from the rate
functions and hold whichever way ties are resolved. Tables~\ref{tab:p1} and
\ref{tab:hetero} are the only experiments that scan the power cap; every other
run is uncapped, the regime $C \ge N$ in which Theorem~\ref{th:nbody} holds. The
predictions of
\S\ref{sec:pair}--\S\ref{sec:nbody} were written down (\texttt{PROTOCOL.md} in
the replication package) after an exploratory fixed-step sweep and before the
confirmatory runs were inspected. We call this a written protocol and not a
pre-registration: the file carries no third-party timestamp, so it documents the
order in which we worked and does not certify it, and two of its predictions
were subsequently corrected by the analysis rather than by the data.

A random launch draws each job's initial progress through its compute phase
uniformly on $[0,1)$, so first write starts are uniform on $[0,T]$ and not on the
cycle $[0,P_0]$; the probability that such a launch is collision-free is
$(1-(N-1)d)^N$, which is $0.361$ at $N=4$, $L=0.3$ against $0.375$ for the
uniform-over-the-cycle law of \eqref{eq:frozenmeasure}, and it is the former that
governs the rejection rates reported here. Launches are conditioned on the
complement of $\mathcal{C}$. A launch inside it is rigid for ever
(Proposition~\ref{prop:frozen}) and returns an exactly zero separation rate, so
averaging the two populations together reports neither; the number of draws
rejected is given per cell in Table~\ref{tab:lyap}. Stochastic quantities are
reported as mean $\pm$ standard error over the retained launches, $60$ per cell
for the separation rate, $20$ for the configuration and moment tables and $500$
for the first-passage times of Table~\ref{tab:life}, with the spread across
seeds given where it matters.

Phases are defined as follows. Let $s_i^k$ be the instant of the $k$-th write
start \emph{of job $i$}, so that a job falling behind is not re-indexed; let $P$
be the median of $s_i^{k+1} - s_i^k$ over the retained window; and set
$\phi_i^k = 2\pi (s_i^k - \bar{s}^k)/P \bmod 2\pi$, with $\bar{s}^k$ the mean
over jobs. Unless stated otherwise the first half of each run is discarded.

\subsection{Pair neutrality and the \texorpdfstring{$N$}{N}-writer spectrum}

Table~\ref{tab:p1} reports the drift of the firing gap over $120$ cycles for two
identical jobs. Table~\ref{tab:spec} checks Theorem~\ref{th:nbody} and
Proposition~\ref{prop:anon} together: the Jacobian of the return map, measured
by central differences, reproduces the predicted spectrum, and the determinant
stays at $1$ while individual eigenvalues move by more than a factor four
between throughput rules. For $N=3$ the measured Jacobian in firing-time
coordinates is $\bigl(\begin{smallmatrix} 2 & 0 \\ 3/2 & 1/2
\end{smallmatrix}\bigr)$ to six digits, and the spectrum is independent of the
configuration within the fully overlapping region, as \eqref{eq:nmap} requires.
Figure~\ref{fig:spectrum} plots the work-conserving case across four fleet sizes.

\begin{figure}[t]
\centering
\includegraphics[width=0.52\textwidth]{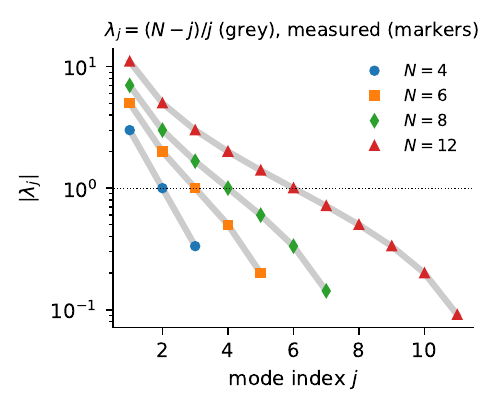}
\caption{Measured against predicted spectrum $\lambda_j = (N-j)/j$ (grey) of the
return map for $N$ fully overlapping writers under the work-conserving rule, at
one random configuration per fleet size; markers are the eigenvalue moduli of the
Jacobian by central differences. Their reciprocal pairing about the dotted unit
line is the unit determinant of Proposition~\ref{prop:anon}, and the
$\lceil N/2 \rceil - 1$ markers above it the expanding directions of
Theorem~\ref{th:nbody}.}
\label{fig:spectrum}
\end{figure}

\begin{table}[t]
\centering\small
\begin{tabular}{lll}
\toprule
Regime & Parameters & $\max |\delta_{120}-\delta_0|$ \\
\midrule
Storage only        & $d \in \{0.05,0.2,0.45\}$, $C=\infty$ & $1.1\times10^{-15}$ \\
Power cap binding   & $C \in \{1.5, 1.05\}$, $d = 0.05$     & $8.9\times10^{-15}$ \\
Mixed               & $C \in \{1.5,1.05\}$, $d \in \{0.2,0.45\}$ & $9.8\times10^{-15}$ \\
Severe cap, $C<1$   & $C \in \{0.9,0.7,0.5\}$, $d = 0.1$    & $1.2\times10^{-14}$ \\
Beyond the proof, $d>T$ & $d \in \{1.5, 2.5\}$, $C \in \{\infty, 1.5\}$ & $1.1\times10^{-15}$ \\
Unequal volumes     & $d_1=0.1$, $d_2=0.3$ & increment $0.2$ on $82/106$ cycles \\
\bottomrule
\end{tabular}
\caption{Pair neutrality, three initial gaps per cell spanning the cycle. Drifts
are at the level of floating-point noise, confirming Theorem~\ref{th:pair}
across all three coupling regimes. The fifth row is outside the hypothesis
$d < T$ (Remark~\ref{rem:dT}); the last confirms the translation term of
Proposition~\ref{prop:hetero} on the cycles whose writes are isolated, the other
$24$ carrying a delay shared equally but falling on cycles of different index.}
\label{tab:p1}
\end{table}

\begin{table}[t]
\centering\small
\begin{tabular}{lccc}
\toprule
throughput rule $f(n_w)$ & $N$ & $\det J$ & $\max_j |\lambda_j^{\text{meas}} - \lambda_j|$ \\
\midrule
$1$ (work-conserving)      & $3,4,5,8,12$ & $1.00000$ & $3.9\times10^{-9}$ \\
$n^{-0.4}$ (degrading)     & 4, 6 & $1.00000$ & $5.6\times10^{-9}$ \\
$1/(1+\tfrac12(n-1))$      & 4, 6 & $1.00000$ & $4.9\times10^{-9}$ \\
$1 + 0.3\ln n$ (improving) & 4, 6 & $1.00000$ & $4.1\times10^{-9}$ \\
$2$ if $n$ even, else $1$  & 4, 6 & $1.00000$ & $2.8\times10^{-9}$ \\
\bottomrule
\end{tabular}
\caption{Measured versus predicted spectrum $(N-j)f(j)/\bigl(j f(N-j)\bigr)$ of
the return map for $N$ fully overlapping writers. For $N = 6$ the leading
eigenvalue ranges from $3.37$ to $15.0$ across these five rules and the
determinant is $1$ in every case; the last rule is not monotone in $n_w$, so no
smoothness or ordering of $f$ is used. The map is affine here, so the residuals
are the round-off floor of the difference quotient at $\varepsilon = 10^{-7}$.}
\label{tab:spec}
\end{table}

\subsection{Separation of nearby configurations}
\label{sec:lyap}

Away from full overlap the dynamics is piecewise affine, and the combinatorics
of who overlaps whom changes as intervals expand and contract. We measure the
finite-time separation rate per cycle from the divergence of two trajectories
initially $10^{-9}$ apart in one job's phase, fitted over the window in which
the separation stays below $10^{-3}$ (Table~\ref{tab:lyap}). Separation means
$\max_i |s_i^k - \tilde{s}_i^k|$, the largest discrepancy between corresponding
write starts of the two trajectories, on absolute firing times: neither reduced
modulo the cycle nor quotiented by the rotation that \eqref{eq:dist} divides
out, both of which would bound a quantity we want to see grow. We call this a
separation rate and not a Lyapunov exponent: it is a finite-time,
finite-separation fit along one pair of trajectories per seed, in one
perturbation direction, and no tangent-space calculation with jump matrices at
the event surfaces is attempted here. Figure~\ref{fig:lyap} shows one such pair
of trajectories per cell.

\begin{figure}[t]
\centering
\includegraphics[width=0.52\textwidth]{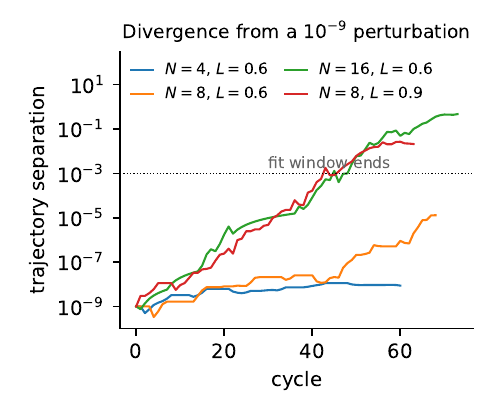}
\caption{Separation of two trajectories launched $10^{-9}$ apart in one job's
phase, one contending launch per cell, uncapped. Over the cycles plotted the
separation gains $8.7$ decades at $N=16$, $L=0.6$ and $7.3$ at $N=8$, $L=0.9$,
against $0.9$ at $N=4$, so the growth rises with fleet size and with load. The
rates fitted on these four traces, $0.035$ to $0.299$ per cycle in the order of
the legend, depart from their cell means (Table~\ref{tab:lyap}) by up to a factor
$2.3$: the spread within a cell is wide, so a single trace carries the trend in
$(N, L)$ and not the mean. Each curve ends at the last cycle every job has
completed within the horizon, and the fits use the window below the dotted line.}
\label{fig:lyap}
\end{figure}

Three controls. The perturbed job does not matter: perturbing each in turn, over
a fixed subsample of $20$ launches, gives per-job means in $[0.184, 0.197]$ at
$N=8, L=0.6$ and $[0.297, 0.318]$ at $N=16, L=0.9$, straddling the $0.188$ and
$0.305$ that subsample returns when job $1$ is perturbed. The spread across launches is wide, the $60$ retained launches
spanning $[0.041, 0.284]$ and $[0.196, 0.461]$ in those two cells, where it is
one-signed. At $L = 0.3$, and at $N = 4$ throughout, it is not: up to a sixth of
the contending launches return a rate that is zero to machine precision and
three return a small negative one (Table~\ref{tab:lyap}), so nearby trajectories
separating is a statement about most launches and not about all of them. Those
null launches have a mechanism and not an exception. Every one of them spends at
most $4\times10^{-5}$ of its time with three concurrent writers, and exactly none
in four of the five cells where such launches occur at all, so every collision
along such a trajectory is a two-body collision and Theorem~\ref{th:pair} leaves
nothing to amplify. That
condition is necessary and not sufficient, chains of two-body collisions not
decomposing into independent pairs: the launches that also stay two-body but do
separate return rates confined to $[0.037, 0.053]$ across every cell, a
population well below the cell means and distinct from it. The
fit window does matter, and is the dominant uncertainty. On the same subsample the
reported window gives $0.188$ against the $0.200$ of the full $60$, while
$[10^{-6}, 10^{-3}]$ and $[10^{-4}, 10^{-2}]$ give $0.289$ and $0.251$ at
$N=8, L=0.6$, a spread of $50\%$; those two are reached within $80$
cycles by $9$ and by $1$ of the $20$ launches, so they are biased
upward and the second is one trajectory. The value of $\lambda$ should be read
to within a factor of about
$1.5$, and only its sign and its trend in $(N, L)$ are used below. That trend in
$L$ carries a further caveat: the rejection rule removes the collision-free
draws, whose share falls as $L$ grows ($28$ rejections at $N=4$, $L=0.3$ against
none at $L=0.9$), so the launches compared along a row are not drawn from the
same conditional law and part of the increase is a change of population rather
than of dynamics.

\begin{table}[t]
\centering\small
\begin{tabular}{rccc}
\toprule
$N$ & $L=0.3$ & $L=0.6$ & $L=0.9$ \\
\midrule
4  & $+0.072 \pm 0.007$ $(28\,|\,10)$ & $+0.081 \pm 0.007$ $(8\,|\,5)$ & $+0.084 \pm 0.008$ $(0\,|\,5)$ \\
8  & $+0.110 \pm 0.008$ $(1\,|\,9)$ & $+0.200 \pm 0.006$ $(0\,|\,0)$ & $+0.227 \pm 0.008$ $(0\,|\,0)$ \\
16 & $+0.127 \pm 0.008$ $(0\,|\,5)$ & $+0.266 \pm 0.006$ $(0\,|\,0)$ & $+0.305 \pm 0.008$ $(0\,|\,0)$ \\
\bottomrule
\end{tabular}
\caption{Finite-time separation rate $\lambda$ per cycle, $80$ cycles, uncapped,
$60$ launches per cell conditioned on the complement of $\mathcal{C}$, mean $\pm$
standard error. In parentheses: the number of draws rejected as collision-free,
and after the bar the number of retained launches whose rate is null
($|\lambda| < 10^{-4}$), which \S\ref{sec:lyap} shows are the launches on which
no instant carries three writers. Three launches at $N = 4$ return a rate below
$-10^{-4}$, the most negative being $-0.007$. The rate grows with load and with
fleet size and does not saturate between $N = 8$ and $N = 16$, so we do not
extrapolate it. The control is Proposition~\ref{prop:frozen}: started from an
evenly staggered collision-free configuration, the same measurement returns a
separation that does not grow at all.}
\label{tab:lyap}
\end{table}

\subsection{The order sector, and no memory resolved beyond it}
\label{sec:order}

Separation of nearby trajectories says nothing about where a single trajectory
goes. Four measurements bound it, in the four cells of Table~\ref{tab:conf} and on
the timescales stated; we have not scanned the cap, the intermediate loads, or
launches built to be clustered.

\emph{The configuration keeps its order, and nothing resolvable beyond it.}
Table~\ref{tab:conf} reports the per-cycle displacement of a job's relative
phase, its standard deviation over the run, the distance $D$ to the launch
configuration and the configuration autocorrelation
$\rho(\ell) = \langle \cos(\phi_i^{k+\ell} - \phi_i^k) \rangle$, with
\begin{equation}
D(\phi, \psi) := \frac{1}{2\pi} \min_{\alpha \in [0, 2\pi)}
\Bigl[\frac1N \sum_{i=1}^{N} \bigl|\phi_i - \psi_i - \alpha\bigr|_{2\pi}^2
\Bigr]^{1/2} ,
\label{eq:dist}
\end{equation}
where $|\cdot|_{2\pi}$ is the representative in $(-\pi, \pi]$ and the rotation
$\alpha$ is divided out, that rotation being the time-translation symmetry
of an autonomous system and not a reorganisation. The minimum is taken exactly,
over the $N$ candidates that make the piecewise-quadratic objective stationary,
and not at the circular mean $\arg \sum_i e^{\mathrm{i}(\phi_i - \psi_i)}$, which
minimises the chordal sum instead and overstates $D$ by up to $0.04$ of a cycle
at these fleet sizes. Phases move by $1$--$6\%$ of a
cycle per cycle and fluctuate within $5$--$10\%$ of a cycle of their own mean.
The distance to the launch settles at $7$--$14\%$ of a cycle, against the
$22$--$26\%$ separating two independent configurations, which reads as a fleet
staying nearer its launch than chance would put it.

That reading is wrong, and the reason is Proposition~\ref{prop:order}. Two
independent configurations generally differ by a permutation of the jobs, and
the cyclic firing order is a constant of the motion, so no trajectory can ever
produce one from the other: the $22$--$26\%$ measures a rearrangement the
dynamics is forbidden to perform. The reference has to be drawn in the order
sector the launch fell in. Taking it as an independent launch of the same law,
relabelled into that sector, which is a benchmark for the complete loss of
detectable memory \emph{inside} the sector and not the farthest a reorganisation
could carry the fleet, gives $0.124 \pm 0.007$, $0.138 \pm 0.004$,
$0.099 \pm 0.004$ and $0.074 \pm 0.003$ in the four cells of
Table~\ref{tab:conf}, against $0.122$, $0.136$, $0.098$ and $0.074$ measured
(\texttt{sim/nullorder.py}). Equality of two means is not equality of two laws,
so that comparison is run three ways. The paired difference is $-0.0013$,
$-0.0018$, $-0.0003$ and $+0.0001$ of a cycle, and against a declared
equivalence margin of $0.01$, a tenth of the separation at issue, $20$ launches
establish equivalence in the two larger fleets and bound the difference by
$0.017$ at $N = 8$. The two pooled distributions of $D$ differ by a
$1$-Wasserstein distance of $0.0005$ to $0.0029$ of a cycle against an
interquartile spread of $0.034$ to $0.075$, with $p \ge 0.74$ under a
permutation that swaps the two references launch by launch. The autocorrelation
behaves the same way: the
$0.64$--$0.89$ plateau of Table~\ref{tab:conf} is $0.63$--$0.89$ for the
relabelled reference. What the fleet retains of its launch is its firing order,
which is a theorem; beyond it these observables resolve nothing, at the power
just stated and not as a proof of independence.

The dynamics is therefore neither frozen, which would give $\rho \equiv 1$ and zero
displacement, as the collision-free control does exactly, nor free to mix: the
autocorrelation is already at its sector value by a lag of $50$ cycles and is
still there at $150$. That also bounds what the separation rate of
\S\ref{sec:lyap} can mean, the sector being invariant, so exponential separation
saturates at the scale of that sector and not of the cycle.

\emph{Ranks never change, and no gap is seen to close.} In every cell of
Table~\ref{tab:conf}, over $20$ launches and $150$ cycles, the cyclic firing
order is the same at every cycle, as Proposition~\ref{prop:order} requires.
Invariance of the order is compatible with all the gaps contracting towards
zero without ever crossing, which is what asymptotic clustering would look like,
so we follow the gaps by rank, which that invariance is what makes possible:
rank $j$ separates the same pair of jobs at every cycle. The minimum over ranks
is stationary. Its mean over launches changes by $+21\%$, $+4\%$, $-24\%$ and
$+44\%$ over $800$ cycles, from $0.0137 \pm 0.0025$ to $0.0166 \pm 0.0022$ cycles
at $N=8$, $L=0.3$ and from $0.0009 \pm 0.0002$ to $0.0013 \pm 0.0003$ at $N=32$,
$L=0.9$, and paired launch by launch its ratio has a $95\%$ interval inside
$[0.37, 2.9]$ in all four cells: a contraction by more than a factor $2.7$ is
excluded over that span, a slower one is not. No particular pair tightens,
though. The gap of the rank smallest at launch has a paired ratio confined to
$[2.2, 68]$, above $1$ in every cell, and the rank realising the minimum has
moved by the end of the run in $16$ to $20$ of the $20$ launches: the small gaps
are a rotating population and not a nucleating cluster.

\emph{No coherence in the first four Daido moments.} Assign
$R_m = |\langle e^{\mathrm{i} m \phi} \rangle|$ \citep{daido1996} for
$m = 1, \dots, 4$: the first moment alone would not settle the question, since
two antipodal clusters, three equidistant ones and a splay state all give
$R_1 \simeq 0$ and are distinguished by the higher moments, with $R_m = 1$ at
$m$ equal to the number of clusters. The reference is not $0$ but the finite-$N$
value for independent uniform phases, which we resample at the same $N$ rather
than take from the Rayleigh limit $\sqrt{\pi}/(2\sqrt{N})$; the two differ by
$0.5\%$, which is not negligible against the departures under test. The largest
departure anywhere in the $36$ cells of Table~\ref{tab:order} is $+3.2$ standard
errors ($N=8$, $L=0.9$, $m=3$), and multiplicity alone does not dispose of it.
The $36$ statistics are not independent, four moments being read off the same
configurations and three loads sharing their launch draws at fixed $N$;
sign-flipping whole launches, which resamples the joint null with that
dependence in place (median absolute correlation $0.20$), puts the $95$th
percentile of the maximum at $2.80$ rather than the $3.72$ that $36$ independent
Student statistics on $19$ degrees of freedom give. What that departure does not
survive is the choice of null. Drawing phases under the launch law of the runs,
uniform on the compute phase and conditioned outside $\mathcal{C}$, raises the
floor at that cell from $0.315$ to $0.328$ and brings the departure down to
$+1.7$ standard errors; and the paired test, which differences $R_m$ over the run
against $R_m$ at the launch of the same run and so cancels the launch law
exactly, returns $-0.002 \pm 0.042$ there, its largest departure anywhere being
$+2.3$ standard errors ($N=16$, $L=0.3$, $m=3$) against a jointly resampled
threshold of $2.84$. Retaining the third and fourth quarters of the run
separately changes no conclusion. What these moments exclude is a low-order
cluster state or a splay; they are not precise enough to exclude a weak one, and
a cell $3.2$ standard errors above the uniform floor is the size of departure
they leave open.

\emph{The burst distribution is nonetheless not that of independent phases, and
its tail is the heavier one.} The
quantity an operator sees is the number of jobs writing at the same instant.
Conditioned on the effective write duty $q$ the fleet realises, independent
phases would give $\mathrm{Binomial}(N, q)$, and the measured distribution
departs from it in both directions. At $N=8$, $L=0.3$ ($q = 0.046$) the fleet
spends $7.0 \pm 0.8\%$ of its time with two or more writers against $5.0\%$ for
independent phases; at $N=32$, $L=0.9$ ($q = 0.100$) it spends
$63.9 \pm 1.3\%$ against $84.5\%$, so ordinary crowding is rarer there. The upper
tail goes the other way in every cell: the $99$th percentile of the number of
concurrent writers is $3$, $5$, $5$ and $12$ against $2$, $4$, $4$ and $8$ for
independent phases, and the $99.9$th is $4$, $6$, $6$ and $15$ against $3$, $6$,
$5$ and $9$. A fleet that does not lock can still spend a thousandth of its time
above the worst concurrency independent phases would ever produce, which is why
\S\ref{sec:practice} keeps the two claims apart. We compare time-weighted
quantiles of the two laws and not
sample maxima, which depend on how long each is observed. The departure is not
an artefact of the launch: the same statistic on the launch configuration itself,
before any dynamics, sits within $0.006$ of its binomial value in all four cells,
against departures of $+0.020$ to $-0.206$ after $300$ cycles. Part of what
remains is mechanical rather than dynamical, a write window lengthening precisely
when it overlaps others, which inflates the time at high $n_w$ at fixed duty.

These four say what the fleet does in these cells: it wanders over the sector
its launch fixed, without freezing and without leaving it, and it does not
cluster. That is why the moments sit at the incoherent floor. A random launch
produces an incoherent configuration and the dynamics does not transform it,
which is a statement about the map and not evidence that the invariant measure
is uniform; the writer distribution shows directly that ``no synchronisation''
is not the same as ``as if independent''. In particular the moments do not
license calling the long-run phase distribution independent, and we do not.

\begin{table}[t]
\centering\small
\begin{tabular}{rrcccccccc}
\toprule
$N$ & $L$ & step/cycle & spread & $\bar{D}$ & $\bar{D}_{\text{sector}}$ & $\rho(1)$ & $\rho(50)$ & $\rho(150)$ & rank ch. \\
\midrule
8  & 0.3 & $0.011$ & $0.085$ & $0.122$ & $0.124$ & $0.996$ & $0.689$ & $0.709$ & $0$ \\
8  & 0.9 & $0.064$ & $0.104$ & $0.136$ & $0.138$ & $0.887$ & $0.645$ & $0.645$ & $0$ \\
16 & 0.6 & $0.029$ & $0.073$ & $0.098$ & $0.099$ & $0.975$ & $0.800$ & $0.799$ & $0$ \\
32 & 0.9 & $0.045$ & $0.055$ & $0.074$ & $0.074$ & $0.942$ & $0.888$ & $0.886$ & $0$ \\
\midrule
\multicolumn{2}{l}{collision-free control} & $8\times10^{-16}$ & $0$ & $2\times10^{-13}$ & --- & $1$ & $1$ & $1$ & $0$ \\
\bottomrule
\end{tabular}
\caption{Motion of the configuration, uncapped, $800$ cycles, second half
retained, $20$ launches. ``step/cycle'' is the mean per-cycle change of a job's
relative phase and ``spread'' its standard deviation over the window; $\bar{D}$
is the mean distance \eqref{eq:dist} to the launch and
$\bar{D}_{\text{sector}}$ the same distance to an \emph{independent} launch
relabelled into the trajectory's order sector (\texttt{sim/nullorder.py}), both
in units of the cycle. Standard errors over launches are at most $0.007$ on the
distances, $0.005$ on the first two columns and $0.031$ on the correlations. Two
independent configurations whose relative order is left free are $0.220$,
$0.241$ and $0.256$ apart at $N = 8$, $16$ and $32$: the reference the
fleet appears to beat, and does not. The largest distance reached within a run
averages $0.206$, $0.257$, $0.194$ and $0.175$, so a trajectory does visit
configurations farther from its launch than a fresh draw in its own sector, one
reason that draw is a benchmark and not a bound. ``rank ch.'' counts cycles whose
cyclic firing order differs from the previous one, over all launches, on a
separate $300$-cycle run.}
\label{tab:conf}
\end{table}

\begin{table}[t]
\centering\small
\begin{tabular}{rrccccc}
\toprule
$N$ & $L$ & $R_1$ & $R_2$ & $R_3$ & $R_4$ & floor \\
\midrule
8  & 0.3 & $0.306 \pm 0.017$ & $0.313 \pm 0.014$ & $0.312 \pm 0.015$ & $0.317 \pm 0.016$ & $0.315$\\
8  & 0.6 & $0.314 \pm 0.014$ & $0.332 \pm 0.013$ & $0.324 \pm 0.009$ & $0.322 \pm 0.010$ & $0.315$\\
8  & 0.9 & $0.319 \pm 0.018$ & $0.332 \pm 0.014$ & $0.343 \pm 0.009$ & $0.331 \pm 0.007$ & $0.315$\\
\midrule
16 & 0.3 & $0.216 \pm 0.007$ & $0.244 \pm 0.012$ & $0.218 \pm 0.007$ & $0.208 \pm 0.006$ & $0.222$\\
16 & 0.6 & $0.225 \pm 0.006$ & $0.217 \pm 0.005$ & $0.221 \pm 0.005$ & $0.220 \pm 0.006$ & $0.222$\\
16 & 0.9 & $0.233 \pm 0.012$ & $0.236 \pm 0.011$ & $0.232 \pm 0.008$ & $0.228 \pm 0.006$ & $0.222$\\
\midrule
32 & 0.3 & $0.154 \pm 0.011$ & $0.162 \pm 0.009$ & $0.171 \pm 0.007$ & $0.157 \pm 0.007$ & $0.157$\\
32 & 0.6 & $0.170 \pm 0.008$ & $0.166 \pm 0.006$ & $0.163 \pm 0.006$ & $0.163 \pm 0.005$ & $0.157$\\
32 & 0.9 & $0.176 \pm 0.010$ & $0.168 \pm 0.007$ & $0.163 \pm 0.004$ & $0.160 \pm 0.004$ & $0.157$\\
\bottomrule
\end{tabular}
\caption{Daido moments below the stagger-feasibility threshold, uncapped, $300$
cycles, second half retained; mean $\pm$ standard error over $20$ launches, no
draw having been rejected at these fleet sizes. The floor is $\mathbb{E}[R_m]$
for $N$ independent uniform phases, resampled over $2\times10^5$ draws (the
Rayleigh approximation $\sqrt{\pi}/(2\sqrt{N})$ gives $0.313$, $0.222$, $0.157$).
Under the launch law of the runs it rises to $0.324$ and $0.328$ in the two
$N = 8$ cells discussed in the text, which brings the largest departure in the
table down to $+1.7$ standard errors. A $k$-cluster state
would show $R_k \to 1$, a splay state every $R_m$ well below the
floor; neither appears at any $(N, L, m)$. Moments up to $m = 4$ do not resolve
five or more clusters, unequal or unequally spaced ones, or associations between
particular jobs, so this is a test against low-order coherence and not against
all of it.}
\label{tab:order}
\end{table}

\paragraph{What the load parameter does not show.}
\label{sec:trap}
Sweeping the duty $d$ upward without bounding the demand makes $R_1$ rise
steeply, which invites reading a synchronisation transition. It is queueing.
Past $L^\star$, which is $1.03$ at $N = 32$, the fabric cannot serve the fleet
within its free period and the effective cycle stretches to $\approx Nd$: the
integrator gives $P_{\rm eff} = 1.22,\ 3.46,\ 6.78,\ 11.31$ at
$L = 1.05,\ 3.2,\ 6.4,\ 11.2$, while $R_1$ climbs from $0.28$ to $0.61$, $0.67$
and back to $0.48$. Every job now spends most of its time writing, so a high
$R_1$ records queue occupancy, and its non-monotonicity in $L$ confirms that it
is not measuring coherence. That leaves the synchronisation question meaningful
above $L^\star$ and makes $R_1$ the wrong observable there: any empirical claim
of fleet synchronisation must control for $L$.

\section{What this says operationally}
\label{sec:practice}

\paragraph{In the uncapped homogeneous regime tested here, mutual contention does
not phase-lock a fleet, which is not the same as leaving it quiet.} Symmetric
contention between identical jobs has exactly zero pairwise coupling
(Theorem~\ref{th:pair}), cannot permute the firing order
(Proposition~\ref{prop:order}), and has a three-body term that is volume
preserving on the branch we can compute, with synchrony an unstable fixed point
of it (Corollary~\ref{cor:spec}). Exact synchrony is therefore unreachable in
finite time, but that leaves open an asymptotic approach to it, gaps shrinking
without ever crossing, which happens on the cluster manifolds
(Corollary~\ref{cor:recruit}) that no such fleet can enter, and which we exclude
elsewhere by measurement and not by proof: no cell of \S\ref{sec:order} resolves
a decay of the smallest gap over $800$ cycles, the Daido moments stay within a
change of null of their floor, and the firing order never changes. What those
measurements do not support, once the invariance of that order is taken into
account, is any memory of the launch beyond the order itself. Two restrictions
travel with all of this. It concerns phase locking and not concurrency: the same
runs put the upper tail of the number of concurrent writers \emph{above} its
independent-phase value in every cell (\S\ref{sec:order}), so a fleet with no
aligning mechanism can still burst harder than uncorrelated jobs would, and how
to size for that is not a question settled here. And it concerns this model on
the timescales simulated with the cap not binding: correlated bursts observed in
production would then have to come from mechanisms it excludes. Those are
external common causes, structural rather than dynamical
(jobs launched together, checkpoint intervals set to the same round number of
steps, wall-clock-aligned policies, correlated restarts after a shared failure);
endogenous mechanisms excluded by assumption (asynchronous checkpointing,
unequal allocation, heterogeneity behind a binding cap, delayed power control
\citep{lerouxtardif2026}); and the within-job alignment of thousands of ranks.
The falsifiable content is a redirection: storm incidence should track
launch-time and interval-setting statistics rather than the fleet's history of
past collisions.

\paragraph{A stagger is permanent while the cap does not bind, and jitter is what
ends it.} An offset schedule in which every pair of consecutive write starts is
at least $d$ apart never degrades, and exists whenever $L \le N/(N-1)$. That is
Proposition~\ref{prop:frozen}, and it assumes $C \ge N$: what a binding cap does
to a staggered fleet is open (\S\ref{sec:limits}), and it is the assumption a
production fleet is least likely to satisfy. Under Assumption~\ref{as:determ}
nothing erodes the stagger. Relaxing that assumption is what gives the
intervention a finite lifetime, and the mechanism is diffusive rather than
dynamical: with per-cycle jitter $T_i^k = T(1 + \sigma z_i^k)$, the $z_i^k$
standard normal and independent across jobs and cycles, the work floored at
zero, a fleet inside $\mathcal{C}$ has every job at its own free period, so each
gap performs a \emph{driftless} random walk of per-cycle variance $2\sigma^2$.
The $N$ walks are not independent: adjacent gaps share a jitter draw, so their
increments have covariance $-\sigma^2$, and the sum of the gaps is pinned to the
cycle. The lifetime of a stagger of margin $m$ is then a first-passage time of
that correlated family, which still scales as $(m/\sigma)^2$ in the margin. Table~\ref{tab:life} measures it over
$500$ launches per cell: the median number of cycles to the first overlap is
$0.11$ to $0.24$ times $(m/\sigma)^2$, over two decades in $(m/\sigma)^2$ and a
factor four in $N$: thirteen of the sixteen cells of Figure~\ref{fig:life} fall
in a band of unit slope on log axes although $m$ and $\sigma$ vary separately by
factors of $4$ and $8$, which is the content of the scaling, and the three above
it sit on the one-cycle floor of the measurement. That coefficient is not a
constant of the model: it falls with $N$ and depends on the jitter law through
more than its variance, so it summarises the cells tested and is not a formula.

What the quadratic scaling tests is narrower than it appears. Before the first
overlap the fleet is inside $\mathcal{C}$, where no two writes contend, so
Theorem~\ref{th:pair} is not what keeps the walk driftless there: each gap is a
difference of two independent jitters and would diffuse the same way whatever
the dynamics does after a collision. What the exponent does exclude is a drift
in the gaps under jitter alone, which Proposition~\ref{prop:frozen} does not
cover, being a statement about the deterministic flow. A systematic drift $\mu$
would give a passage time of order $m/|\mu|$ rather than $(m/\sigma)^2$, and
none is seen over two decades.

This replaces the natural but incorrect recipe of re-staggering on the e-folding
time $\ln(1/\varepsilon)/\lambda$ of the separation rate, which is the time for
two \emph{nearby trajectories} to diverge and not, as \S\ref{sec:order} shows,
the time for one trajectory to lose its structure. The operator's quantity is
the first passage of the margin to zero, set by the jitter budget and not by
$\lambda$. A median is not a guarantee, so the useful form is a survival
probability. At $N=8$, $L=0.3$, where an even stagger affords $m = 0.092$, a
jitter of $1\%$ leaves the schedule intact for $10$ cycles with probability
$0.78$ and for $50$ with probability $0.01$; halving the jitter raises those to
$1.00$ and $0.60$. These are proportions over $500$ launches, so their binomial
standard error is at most $0.022$. Buying a hundred cycles at $1\%$ jitter would take a margin
of $0.24$ of a cycle, which no fleet of more than four jobs can offer, since
$m = P_0/N - d$ falls as $1/N$ and is already $0.045$ at $N=16$ and $0.022$ at
$N=32$. At scale it is the refresh interval and not the margin that is the
adjustable quantity.

\paragraph{Fleet size cuts both ways.} The leading expansion rate on the fully
overlapping branch is $N-1$, and the measured separation rate rises with $N$
throughout Table~\ref{tab:lyap}. Working against that, the collision-free set
survives to arbitrary $N$ but on ever tighter terms: its threshold falls to $1$
and the margin per job as $1/N$, so its lifetime under jitter falls with the
$m^2 \sim 1/N^2$ of the margin, and the measured coefficient decreases with $N$
on top of that. Three fleet sizes, three of whose cells sit on the resolution
floor, do not establish an exponent, so we report the direction and not a
scaling.

\begin{table}[t]
\centering\small
\begin{tabular}{rrcrrrr}
\toprule
$N$ & $L$ & margin $m$ & $\sigma = 0.020$ & $0.010$ & $0.005$ & $0.0025$ \\
\midrule
8  & 0.3 & $0.0922$ & $5\ (0.24)$   & $16\ (0.19)$  & $56\ (0.17)$ & $226\ (0.17)$ \\
8  & 0.6 & $0.0594$ & $2\ (0.23)$   & $7\ (0.20)$   & $25\ (0.18)$ & $92\ (0.16)$ \\
16 & 0.3 & $0.0449$ & $1\ (\ast)$   & $4\ (0.20)$   & $11\ (0.14)$ & $42\ (0.13)$ \\
32 & 0.3 & $0.0222$ & $1\ (\ast)$   & $1\ (\ast)$   & $3\ (0.15)$  & $9\ (0.11)$ \\
\bottomrule
\end{tabular}
\caption{Lifetime of an evenly staggered collision-free schedule under per-cycle
jitter of the compute time: median cycles to the first overlap over $500$
launches, and in parentheses the ratio to $(m/\sigma)^2$. Uncapped, margin
$m = P_0/N - d$, jitter standard normal and independent across jobs and cycles.
Cells marked $\ast$ sit at the one-cycle floor of the measurement, where the
median is not resolved and the ratio is meaningless. Elsewhere it lies in
$[0.11, 0.24]$ and decreases with $N$, the expected effect of taking a minimum
over $N$ walks, correlated through the jitter draws they share. The distribution is wide, as a first-passage time is:
at $N=8$, $L=0.3$, $\sigma=0.005$ the $10$th and $90$th percentiles are $30$ and
$110$ cycles. No run was censored at the simulation budget.}
\label{tab:life}
\end{table}

\begin{figure}[t]
\centering
\includegraphics[width=0.52\textwidth]{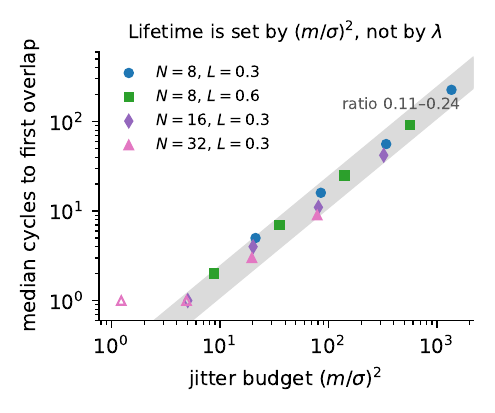}
\caption{The same $500$ launches per cell as Table~\ref{tab:life}, plotted
against the jitter budget $(m/\sigma)^2$; the band is the reported range
$[0.11, 0.24]$ of the ratio. Hollow markers are the three cells whose median sits
at the one-cycle floor, above the band because a median cannot fall below one
cycle and not because the scaling fails there. The residual ordering within the
band is the decrease of the coefficient with $N$.}
\label{fig:life}
\end{figure}

\section{Limitations}
\label{sec:limits}

\subsection{A boundary made explicit: heterogeneity behind a binding cap}
\label{sec:hetero}

Theorem~\ref{th:pair} assumes identical jobs and Proposition~\ref{prop:hetero}
assumes a cap that does not bind. Neither hypothesis can be dropped. Take
$T_1 = T_2 = 1$, $d_1 = 0.2$, $d_2 = 0.4$, a detuning $\Delta = 0.2$, and a cap
$C = 1.2$. Both resources remain anonymous, neither rule being able to name a
job, yet the gap map is not a translation: the integrator gives
$\delta = 0.05 \mapsto 0.283$, $0.10 \mapsto 0.367$, $0.15 \mapsto 0.450$,
against the pure translation $\delta + \Delta$ that
Proposition~\ref{prop:hetero} gives at $C \ge 2$. The increment depends on
$\delta$, which is a coupling in the sense of Definition~\ref{def:anon}.

Two features make it interpretable. The coupling is \emph{supported on gaps
below the detuning}: for $\delta \ge \Delta$ the increment returns to a
constant, $0.333$ in this cell, since the faster writer has then finished and
resumed computing before the slower one is delayed by it. The phase response is
thus carried by an interval of width $\Delta$ in a cycle of length
$\approx 1 + \bar d$, and it disappears as the fleet becomes homogeneous,
consistently with Theorem~\ref{th:pair}. And its strength is set by the gap
between the two compute rates: writing $\kappa := \dd g/\dd\delta$ for the slope
of the increment on that support, Table~\ref{tab:hetero} finds
\begin{equation}
\kappa = \frac{1}{v_2} - \frac{1}{v_1}
\label{eq:kappa}
\end{equation}
to five digits at every cap tested, vanishing exactly at $C = 2$ where
$v_1 = v_2$. We report \eqref{eq:kappa} as a measured law for this
two-parameter family and not as a theorem: the mechanism is clear, the two jobs
converting a work lead into a time lag through different rate sequences once
their write durations differ, but the coefficient is not derived in general, and
neither a fixed point of the resulting map nor its stability is established.
What this bounds is the reach of Theorem~\ref{th:pair}: order-two
phase-blindness is a statement about \emph{identical} units, and a real fleet is
never exactly homogeneous.

\begin{table}[t]
\centering\small
\begin{tabular}{lccrr}
\toprule
cap $C$ & $v_1 = \min\{1,C\}$ & $v_2 = \min\{1,C/2\}$ & $\kappa$ measured & $1/v_2 - 1/v_1$ \\
\midrule
$\infty$, $2.0$ & $1$ & $1$   & $0.0000$ & $0.0000$ \\
$1.99$   & $1$    & $0.995$ & $0.0050$ & $0.0050$ \\
$1.8$    & $1$    & $0.9$   & $0.1111$ & $0.1111$ \\
$1.5$    & $1$    & $0.75$  & $0.3333$ & $0.3333$ \\
$1.2$    & $1$    & $0.6$   & $0.6667$ & $0.6667$ \\
$1.0$    & $1$    & $0.5$   & $1.0000$ & $1.0000$ \\
$0.7$    & $0.7$  & $0.35$  & $1.4286$ & $1.4286$ \\
\bottomrule
\end{tabular}
\caption{Pairwise coupling for two heterogeneous jobs ($d_1 = 0.2$,
$d_2 = 0.4$) behind a binding cap, measured at $\delta = 0.1$ by central
differences on the exact integrator (\texttt{sim/hetero.py}). The coupling is
exactly zero for $C \ge 2$, where the cap never binds with two jobs, and grows as
the two compute rates separate. Anonymity holds throughout: it is not sufficient
for phase-blindness once the units differ.}
\label{tab:hetero}
\end{table}

\subsection{The remaining assumptions, and what is not covered}

\emph{Blocking checkpoints (Assumption~\ref{as:blocking}) is the fragile one.}
Under asynchronous checkpointing the job keeps computing while its state is
flushed, and the cycle length becomes $\max\{T, \text{write duration}\}$. Two
things break at once: the phases stop partitioning time, so \eqref{eq:partition}
fails and with it the step that turns equal writing time into equal computing
time; and the maximum is a threshold nonlinearity, so a job whose write is
stretched past $T$ is delayed while one whose write finishes early is not. We
expect a pairwise coupling to reappear where contention pushes the write past
$T$, but we have not computed it and the expectation is not a result: a first
attempt produced drifts that are integer multiples of the period, that is skipped
checkpoints rather than phase slip, and settling whether a checkpoint that misses
its slot is dropped or queued changes the model before it changes the answer.
This is the first extension to compute, and the regime of modern stacks.

\emph{A binding power cap is outside the $N$-body result.}
Theorem~\ref{th:pair} and Proposition~\ref{prop:order} hold for identical jobs
at every $C > 0$, but Theorem~\ref{th:nbody} needs $C \ge N$: when the cap
binds, jobs that finish writing at different instants compute at rates that
depend on how many others have already finished, and the compute durations cease
to be equal. Measuring the Jacobian numerically at $N = 4$, $d = 0.45$ gives
$\det J = 1.50$ at $C = 2$ and $\det J = 24$ at $C = 1$, against $1$ for
$C \ge N$: on that branch the map is strongly expanding rather than volume
preserving, and the unit determinant of \S\ref{sec:nbody} is a property of
the storage channel alone. Every run in \S\ref{sec:numerics} is uncapped for
this reason, and what a homogeneous fleet does under a binding cap is open. It
is the most consequential gap here, since production fleets are capped.

\emph{Unequal allocation restores coupling; degrading throughput does not.}
Remark~\ref{rem:robust} and Proposition~\ref{prop:anon} cover any anonymous
rule, monotone or not. Priorities, weighted classes and per-job bandwidth caps
break anonymity and are expected to produce a coupling of first order in the
weight difference; we have not derived it.

\emph{Determinism matters because the map is marginally stable in the pairwise
sector}, which is what \S\ref{sec:practice} quantifies: jitter turns the neutral
direction into a random walk, phase gaps diffuse rather than lock, and
collisions recur on the diffusive timescale. For small fleets this competes with,
and may dominate, the separation measured in \S\ref{sec:lyap}.

Four further items are outside scope. \emph{The intra-job burst}: the model
collapses each job to one scalar writer, so a fleet that never synchronises
across jobs can still produce correlated I/O for that reason alone, and
separating the two experimentally means conditioning on the number of
concurrently checkpointing \emph{jobs}, not on aggregate bandwidth.
\emph{Network collectives}, which couple ranks within a job and could set the
effective $T_i$, are absent. \emph{Electrical response}: the facility is not
modelled, so the megawatt-scale figures of \S\ref{sec:intro} are context and not
calibration. \emph{Scale and ergodicity}: the fleets reach $N = 32$, well below
production scale, and no mixing is established in the ergodic-theoretic sense,
which would need a full Lyapunov spectrum, a decay-of-correlations estimate and
a global invariant measure.

\section{Relation to the companion paper, and an open question}
\label{sec:open}

The companion paper \citep{lerouxtardif2026} asks whether co-located training
jobs phase-lock behind a shared power envelope, and answers that they can:
load-dependent throttling is a coupling channel whose sign is set by the lag of
the control loop. Its throttle is a controller with memory, acting on aggregate
demand after a lag $\tau$, so the rate a job receives at $t$ depends on the state
of the fleet at $t - \tau$. The cap here is $v = \min\{1, C/n_c(t)\}$, an
instantaneous function of current occupancy and so the $\tau \to 0$ limit of that
family; in that limit the companion's coupling coefficient, proportional to
$-\sin(m \bar\omega \tau)$ in the $m$-th harmonic, vanishes, and the two
statements agree where their domains meet. What is added here is that the
vanishing is exact rather than leading-order, and that its cause is not the
smallness of a parameter but anonymity together with exclusive phases, through
Lemma~\ref{lem:deficit}: the coupling in \citep{lerouxtardif2026} is a property
of the \emph{controller}, not of the sharing.

Everything above concerns one storage fabric, one power cap and two phases per
cycle, and the mechanism used little of that: Theorem~\ref{th:pair} used that
the resource cannot tell two colliding users apart, that its rate depends on the
present occupancy alone, and that a unit is active on exactly one resource at a
time. Two hypotheses are not decoration and travel with the statement: the
phases must partition time, and the section of Definition~\ref{def:anon} must
recur with a constant offset, which is what $d < T$ buys here.

\begin{conjecture}[Anonymity forbids pairwise coupling]
\label{conj:blind}
Let a population of \emph{identical} integrate-and-fire units interact only
through anonymous resources that are \emph{memoryless}, that is whose rates
depend on the present occupancy alone, in any number and with arbitrary rate
functions, through any number of phases per cycle, with each unit active on
exactly one resource at any instant, and suppose the section of
Definition~\ref{def:anon} recurs along every trajectory with a constant offset.
Then the coupling is phase-blind at order two.
\end{conjecture}

\noindent The conjecture is deliberately confined to order two. Whether the
leading nonvanishing interaction is then of order three, and whether the
three-body map preserves volume on the fully overlapping branch, are separate
questions that do not follow from it and that we do not conjecture.

Two of the three structural hypotheses are known to be necessary rather than
convenient, the recurrence of the section being a regularity condition.
\emph{Identical units}: \S\ref{sec:hetero} exhibits an anonymous, memoryless,
exclusive-phase system with a nonzero pairwise coupling as soon as two units
differ and the cap binds. \emph{Memoryless resources}:
\citep{lerouxtardif2026} exhibits an anonymous resource with a lag whose coupling
is nonzero and whose sign is tunable. \emph{Exclusive phases} we believe
necessary, on the strength of the broken proof step of \S\ref{sec:limits}, but we
have exhibited no counterexample, and a broken proof is not a false
conclusion.

The obstruction is identifiable, which is the reason for stating the conjecture.
The occupancy identity \eqref{eq:occupancy} closes because
Assumption~\ref{as:blocking} makes the two phases partition time: equal total
writing time is read off the write phase and transported to the compute phase,
where it becomes equal solo work and hence an invariant work difference. With
$M$ resources and $p$ mutually exclusive phases the same bookkeeping gives $p$
occupancy identities and one partition identity, and nothing forces the solo
intervals on resource $r$ to match those on resource $r'$ unit by unit: the
single scalar $\Delta x$ must be replaced by an object that survives the
combinatorics of which unit is where. The conjecture also
has a reading testable long before it is settled: \emph{if it holds, a memoryless
anonymous resource cannot steer relative phase}. Such a resource still fixes
everything else, the
absolute rate, the mean period, the duty cycle and the queue, so uncontrollable
is the wrong word for it; what it cannot do is act on the gap between two
identical users as a function of that gap. An operator who wants to steer
relative phases through the resource a fleet shares must therefore break
anonymity, add memory, or exploit the heterogeneity that is already there.

\section{Conclusion}

Treating checkpointing jobs as pulse-coupled oscillators makes a
counter-intuitive prediction that we prove and verify within the model: between
identical jobs, symmetric contention has no pairwise coupling at all, it cannot
permute the firing order, and its three-body coupling preserves volume on the
branch where it can be computed, which makes synchrony a fixed point with
$\lceil N/2 \rceil - 1$ expanding directions rather than an attractor. A cluster
does recruit, but only inside the manifold on which it already exists
(Corollary~\ref{cor:recruit}), which a fleet firing at distinct instants cannot
enter. Sharing one storage fabric therefore gives such a fleet no mechanism,
within this model, that would drive it into phase locking of its own accord,
which is the sense in which its checkpoint storms are not self-reinforcing; its
bursts are not thereby milder than uncorrelated ones, the upper tail of
concurrency being the heavier in every cell measured. What it does instead
depends on how it was started and on very little else: spread widely enough that
no two writes overlap it is rigid, a regime available exactly below the
stagger-feasibility threshold, reachable from nowhere else and conditional on the
cap not binding; started outside it, in the cells we measured, it neither locks
nor clusters over hundreds of cycles, and the only memory of its launch these
observables resolve is the firing order that anonymity freezes.

Whether any of this depends on checkpointing is the open question of
\S\ref{sec:open} rather than a result: the proof uses two exclusive phases, a
deterministic amount of work per cycle and $d < T$, and
Conjecture~\ref{conj:blind} is precisely the claim that the first of those is
not essential. The boundaries are as informative as the statement, and two of the
three hypotheses can be removed with the coupling returning: with memory it
returns with a tunable sign \citep{lerouxtardif2026}, and with heterogeneity
behind a binding cap at $O(1)$ on a support set by the detuning. What follows for
practice is that a stagger is worth constructing, that it is permanent in the
deterministic uncapped model, and that the quantity governing its refresh is the
jitter budget through $(m/\sigma)^2$ rather than any exponent of the free
dynamics.

\paragraph{Reproducibility.} The replication package, submitted as ancillary
files with this preprint, contains the integrator including its tie-breaking
convention, its launch law, its rejection rule and the distance
\eqref{eq:dist} with the grid check of its rotation minimisation
(\texttt{sim/exact.py}), the written protocol (\texttt{PROTOCOL.md}), and one
script per table: \texttt{test\_p1} (Table~\ref{tab:p1}), \texttt{spectrum}
(\ref{tab:spec}), \texttt{dynamics} (\ref{tab:lyap}), \texttt{checks}
(\ref{tab:conf}), \texttt{nullorder} (its order-sector reference and the
equivalence tests), \texttt{order} (\ref{tab:order}), \texttt{burst} (the writer
distribution of \S\ref{sec:order} and the lifetimes of Table~\ref{tab:life}),
\texttt{hetero} (\ref{tab:hetero}), with \texttt{frozen} covering
\S\ref{sec:invariant} and \S\ref{sec:trap}. All runs are seeded, require only
\texttt{numpy}, and \texttt{python sim/all.py} regenerates every number quoted
here, under Python 3.14.0 and numpy 2.4.6. The four figures are produced
separately by \texttt{sim/figures.py}, which additionally requires
\texttt{matplotlib}.

{\footnotesize
\setlength{\bibsep}{3pt plus 1pt}
\bibliographystyle{plainnat}
\bibliography{refs}}

\end{document}